\documentclass[letterpaper,11pt]{article}

\usepackage[normalem]{ulem} 

\usepackage{ifthen}

\usepackage{times} 
\usepackage{fullpage} 
\usepackage{amsmath} 
\usepackage{amssymb} 
\usepackage{amsthm} 
\usepackage{bbm} 

\usepackage{tabularx}

\usepackage{graphicx} 

\usepackage{ifpdf}
\ifpdf
\usepackage[
  pdftex,
  bookmarks,
  pagebackref,
  plainpages=false, 
  pdfpagelabels=true 
]{hyperref}
\else \fi 

\newcommand{\spa}[1]{\mathcal{#1}}

\newcommand{\commentout}[1]{}

\newcommand{\defeq}{\stackrel{\smash{\textnormal{\tiny def}}}{=}}

\newcommand{\eps}{\varepsilon}

\newtheorem{theorem}{Theorem}

\newtheorem{corollary}[theorem]{Corollary}
\newtheorem{proposition}[theorem]{Proposition}

\theoremstyle{definition}

\newtheorem{defn}[theorem]{Definition}

\newtheorem{protocol}[theorem]{Protocol}

\newcommand{\pa}[1]{(#1)}
\newcommand{\Pa}[1]{\left(#1\right)}

\newcommand{\set}[1]{\{#1\}}

\newcommand{\bra}[1]{\langle#1|}

\newcommand{\ket}[1]{|#1\rangle}

\newcommand{\kb}[1]{\ket{#1} \bra{#1}}
\newcommand{\altketbra}[1]{\kb{#1}}
\newcommand{\ketbra}[2]{\ket{#1} \bra{#2}}

\DeclareMathOperator{\trace}{Tr}
\newcommand{\ptr}[2]{\trace_{#1}\pa{#2}}
\newcommand{\Ptr}[2]{\trace_{#1}\Pa{#2}}

\newcommand{\tinyspace}{\mspace{1mu}}
\newcommand{\abs}[1]{|\tinyspace#1\tinyspace|}

\newcommand{\norm}[1]{\lVert\tinyspace#1\tinyspace\rVert}
\newcommand{\Norm}[1]{\left\lVert\tinyspace#1\tinyspace\right\rVert}

\newcommand{\tnorm}[1]{\norm{#1}_{\trace}}
\newcommand{\Tnorm}[1]{\Norm{#1}_{\trace}}

\newcommand{\F}{\mathrm{F}}

\newcommand{\ol}[1]{{\overline{#1}}}
\def\ot{\otimes}

\def\cA{\mathcal{A}}
\def\cB{\mathcal{B}}
\def\cC{\mathcal{C}}

\def\cM{\mathcal{M}}

\def\cX{\mathcal{X}}
\def\cY{\mathcal{Y}}

\newcommand{\half}{\frac{1}{2}}

\usepackage{color,graphicx}

\newcommand{\zo}{\{ 0, 1 \}}

\newcommand{\PB}{P_{\mathrm{Bob}}^{\star}}
\newcommand{\PA}{P_{\mathrm{Alice}}^{\star}}

\newcommand{\cont}[1]{\mathrm{controlled\textrm{-}}#1}

\begin{document}

\title{
Optimal bounds for semi-honest quantum oblivious transfer
}

\author{
  Andr\'e Chailloux$^*$ $\qquad$ Gus Gutoski$^\dagger$ $\qquad$ Jamie Sikora$^\ddagger$ \\[4mm]
  {\small\it
  \begin{tabular}{c}
    {\large$^*$}INRIA Paris-Rocquencourt, SECRET Project-Team,
78153 Le Chesnay Cedex, France \\[1mm]
    {\large$^\dagger$}Perimeter Institute for Theoretical Physics, Waterloo, Ontario, Canada \\[1mm]
    {\large$^\ddagger$}Laboratoire d'Informatique Algorithmique: Fondements et Applications, Universit\'e Paris Diderot,
    Paris, France
  \end{tabular}
  }
}

\date{August 30, 2016} 

\maketitle
 
\begin{abstract}
  Oblivious transfer is a fundamental cryptographic primitive in which Bob transfers one of two bits to Alice in such a way that Bob cannot know which of the two bits Alice has learned.
  We present an optimal security bound for quantum oblivious transfer protocols, in the information theoretic setting, under a natural and {arguably} demanding definition of what it means for Alice to cheat.
  Our lower bound is a smooth tradeoff between the probability $\PB$ with which Bob can guess Alice's bit choice and the probability $\PA$ with which Alice can guess both of Bob's bits given that she learns one of the bits with certainty.
  We prove that $2 \PB + \PA \geq 2$ in any quantum protocol for oblivious transfer, from which it follows that one of the two parties must be able to cheat with probability at least $2/3$.
  We prove that this bound is optimal by exhibiting a family of protocols whose cheating probabilities can be made arbitrarily close to any point on the tradeoff curve.
\end{abstract}
 
\section{Introduction}
 
The rise of quantum information has rekindled interest in information theoretic cryptography---especially in fundamental two-party primitives such as coin flipping, bit commitment, and oblivious transfer.
Without quantum information, any protocol for any of these primitives is completely insecure against a cheating party, and one must assume a bounded adversary in order to realize these primitives with nonzero security.
With quantum information, however, initial results assert only that \emph{perfect} security cannot be achieved \cite{Mayers97,LoChau97,LoChau97a,Lo97,BuhrmanCS12}, leaving a wide range of possibilities for imperfect unconditional security of these primitives.

Interest has therefore concentrated on quantifying the level of security for these primitives that can be achieved by quantum information.
Many non-trivial results followed, and significant insight into the nature and benefits of quantum information has been gained from these studies.
Optimal security bounds are now known for both coin flipping  \cite{Kitaev02,Mochon07,ChaillouxK09} and bit commitment \cite{ChaillouxK11}, but the security of quantum protocols for oblivious transfer has remained an open question.

It is a fascinating fact that different primitives have different security bounds, as each new bound we learn provides another perspective on quantum information and what can be achieved with it.

\subsection{Semi-honest oblivious transfer} \label{sec:intro:weak-OT}

\emph{Oblivious transfer} is a two-party primitive in which Alice begins with a \emph{choice bit} $a\in\zo$ and Bob begins with two \emph{data bits} $x_0,x_1\in\zo$.
The security goals are:
\begin{enumerate}
  \item \emph{Completeness:} If both parties are honest then Alice learns the value of $x_a$.
  \item \emph{Soundness against cheating-Bob:} Cheating-Bob obtains no information about honest-Alice's choice bit $a$.
  \item \emph{Soundness against cheating-Alice:} \label{it:security:cheating-Alice} Cheating-Alice obtains no information about at least one of honest-Bob's two data bits $x_0,x_1$.
\end{enumerate}
The name ``oblivious transfer'' is derived from these requirements:
Bob \emph{transfers} one of two bits to Alice, and is \emph{oblivious} as to which bit he transferred.
It is typical to give priority to the completeness goal and study the extent to which the soundness goals must be compromised in order to achieve it.

Whereas the security goals for both coin flipping and bit commitment are clear and unambiguous, goal \ref{it:security:cheating-Alice} for oblivious transfer admits no simple metric by which to judge the success of cheating-Alice.
It is often the case that cheating-Alice can sacrifice complete information about one of Bob's data bits in exchange for partial information about both, and there is no ``best'' way to allocate that partial information among the bits.
For example, is it better for cheating-Alice to guess the exclusive-OR of Bob's bits with certainty, or to guess \emph{both} his bits with some chance of error?

In this paper we study quantum protocols for oblivious transfer {with perfect completeness} under a natural and {arguably} demanding definition of what it means for Alice to cheat.
For each such protocol we define the symbols
\begin{center}
\begin{tabularx}{\textwidth}{rX}
  $\PB$: & The maximum probability with which cheating-Bob can guess honest-Alice's choice bit $a$. \\
  $\PA$: & The maximum over $a\in\zo$ of the probability with which cheating-Alice can guess $x_\ol{a}$ given that she guesses $x_a$ with certainty.
  (Here $\ol{a}$ denotes the bit-compliment of $a$.)
\end{tabularx}
\end{center}
Observe that $\PA,\PB\geq 1/2$ for every protocol, as a cheating party can do no worse than a random guess. As mentioned previously, classical (non-quantum) protocols are completely insecure, meaning that either $\PA = 1$ or $\PB = 1$ for any such protocol.
By contrast, a consequence of our results is that quantum protocols can achieve both $\PA<1$ and $\PB<1$ simultaneously.

In studies of primitives with inputs such as bit commitment and oblivious transfer (as opposed to primitives \emph{without} inputs such as coin flipping) it is standard to assume, in the information theoretic setting, that an honest party's input bits are uniformly random, so that the probabilities $\PA,\PB$ are taken over uniformly random data bits for Bob and choice bit for Alice, respectively.
In fact, there is an equivalence between oblivious transfer with and without random inputs that maintains cheating probabilities---see, for example, Refs.\ \cite{Schaffner10,ChaillouxKS13-QIC} and the references therein.
See Section \ref{sec:intro:security} for further comments on the security definition.

We henceforth adopt the name \emph{semi-honest oblivious transfer} to refer to the fact that we require security only against cheating-Alices who guess one of honest-Bob's bits with certainty, and that such a cheating-Alice must learn anything that honest-Alice could learn.

{In this paper, we show that a ${2}/{3}$ cheating probability for both Alice and Bob can be achieved and is optimal. Of course, such a cheating probability is usually not suitable for cryptography and therefore, this result should not be seen as an optimal construction for real life cryptography. The interest of this paper is rather to study the limits of oblivious transfer in the information theoretic setting, and to determine precisely what is possible and what is not in  our quantum mechanical world. Notice also that we provide both a lower bound and an upper bound for quantum oblivious transfer. Therefore, if one considers for example that the security definition is too weak, this will only strengthen our lower bound, while weakening the upper bound.}
 
\subsection{Results} \label{sec:intro:results}

We present an optimal security bound for quantum protocols for semi-honest oblivious transfer.

{\sloppy
\begin{theorem}[Lower bound curve] \label{thm:LowerBound}
  In any quantum protocol for semi-honest oblivious transfer it holds that ${2 \PB + \PA \geq 2}$.
\end{theorem}
}

We show that the lower bound curve of Theorem \ref{thm:LowerBound} is optimal by exhibiting a family of quantum protocols for semi-honest oblivious transfer whose cheating probabilities can be made arbitrarily close to any point on that curve.

\begin{theorem}[Upper bound curve] \label{thm:UpperBound}
  Let $m, n \in [1/2,1]$ be any numbers on the line $2n+m=2$.
  For any $\epsilon > 0$ there exists a quantum protocol for semi-honest oblivious transfer with
  \[ \PA \leq m + \epsilon/2 \quad \text{ and } \quad \PB \leq n + \epsilon/4 \]
  so that $2 \PB + \PA \leq 2 + \epsilon$.
\end{theorem}

Taken together, Theorems \ref{thm:LowerBound} and \ref{thm:UpperBound} completely characterize the pairs $(\PB,\PA)$ that can be obtained by perfect quantum protocols for semi-honest oblivious transfer---see Figure \ref{fig:example} below. 

\begin{figure}[htbp] 
   \centering
   \includegraphics[width=4in]{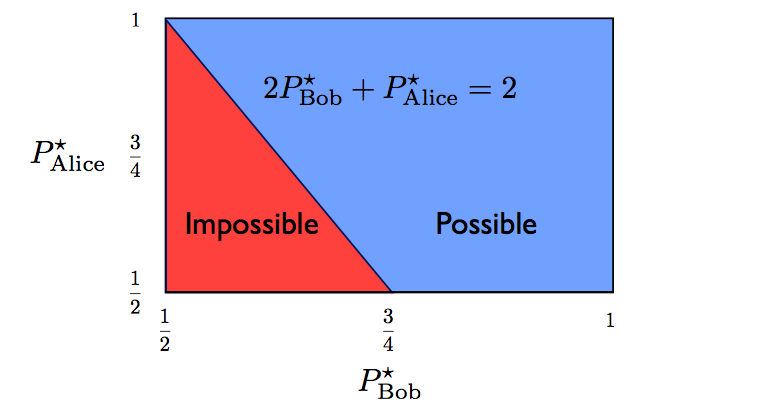}
   \caption{The possible values for $(\PB, \PA)$ in a quantum semi-honest oblivious transfer protocol.}
   \label{fig:example}
\end{figure}

As a corollary, we obtain an optimal bound on the maximum probability with which one party can cheat in any quantum protocol for semi-honest oblivious transfer.

\begin{corollary}[Optimal maximum cheating probability] \label{cor:bound}
  In any quantum protocol for semi-honest oblivious transfer it holds that $\max \set{\PB, \PA} \geq 2/3$.
  Moreover, for any $\epsilon > 0$, there exists a protocol satisfying $\max \set{\PB, \PA} \leq 2/3 + \epsilon$.
\end{corollary}

Thus, there is no hope of ``amplifying'' quantum protocols for semi-honest oblivious transfer in order to get both cheating probabilities close to $1/2$.
 
The security requirements of semi-honest oblivious transfer demand more of cheating-Alice than any previous study of oblivious transfer of which we are aware. 
As such, the lower bound of Theorem \ref{thm:LowerBound} is also an improvement on all known variants of oblivious transfer.
For example:
\begin{itemize}
  \item In Ref.\ \cite{ChaillouxKS13-QIC} it was proven that one party can always cheat with probability $0.585$ under a security requirement similar to ours except that cheating-Alice need not guess one of Bob's bits with certainty.
  \item In Ref.\ \cite{ChaillouxKS14} a cheating probability of $0.599$ was proven under the requirement that cheating-Alice need only guess the exclusive-OR of Bob's bits. 
\end{itemize}
{Cheating with respect to our definition in Section \ref{sec:intro:weak-OT} implies cheating with respect to each of the above definitions.} Therefore, for both of these security definitions, Theorem \ref{thm:LowerBound} improves the cheating probability to $2/3$. 
Further comments on security definitions for oblivious transfer can be found in Sections \ref{sec:intro:security} and \ref{sec:crazy-sec}.

\subsection{Notes on the security definition} \label{sec:intro:security}
 
\begin{enumerate}

\item
  Our definition of semi-honest oblivious transfer has a desirable property: it enforces an intuitively satisfying notion of what it means to ``cheat'' by requiring that cheating parties always learn at least as much as their honest counterparts.
  As mentioned in Section \ref{sec:intro:weak-OT}, cheating-Alice can often sacrifice complete information about one of Bob's bits in exchange for partial information about both.
  Such a strategy has the peculiar property that cheating-Alice knows \emph{less} about one of Bob's bits than honest-Alice!

  Indeed, many attacks on quantum protocols for oblivious transfer exploit this possibility.
  For example, Chailloux, Kerenidis, and Sikora exhibit a protocol in which cheating-Alice can guess both of Bob's bits with probability $3/4$ \cite{ChaillouxKS13-QIC}.
  By contrast, we show in Section \ref{ssect:BobProtocol} that Alice's cheating probability for this protocol drops to $1/2$---indicating that she cannot cheat at all---when we add the requirement that she guess one of Bob's bits with certainty.

\item
  The assumption that an honest party's inputs are uniformly random does not lend itself to proofs of composable security.
  Consequently, we are concerned only with so-called ``stand-alone'' security of semi-honest oblivious transfer in this paper; our results are not known to hold under sequential or parallel composition of multiple protocols.

\item
  It is natural to wonder how the definition of semi-honest oblivious transfer is affected by relaxing it so that cheating-Alice need only guess one of Bob's bits with probability $1-\delta$ for some small $\delta$.
  We comment on this relaxation in Section \ref{sec:crazy-sec}.

\item
Our results are not \emph{cheat-sensitive}, meaning that security is guaranteed even if the honest party detects cheating and chooses to abort the protocol. Each of the cheating strategies we consider in this paper has the property that the honest party is not aware that the dishonest party is cheating.
 
\end{enumerate}

\subsection{Prior work}  \label{sec:intro:prior-work}
  
\emph{Coin flipping} is a two-party primitive in which Alice and Bob wish to agree on a uniformly random bit in such a way that a cheating party cannot bias the sampling distribution of that bit.
\emph{Strong} coin flipping refers to the requirement that a cheater cannot bias the result in either direction, whereas \emph{weak} coin flipping assumes that Alice and Bob only cheat towards opposing outcomes.
Remarkably, there exist quantum protocols for weak coin flipping with cheating probabilities arbitrarily close to $1/2$ \cite{Mochon07}, achieving near-perfect unconditional security.
By contrast, in any quantum protocol for strong coin flipping at least one party can cheat with probability $1/\sqrt{2}\approx 0.707$~\cite{Kitaev02} and this bound is optimal~\cite{ChaillouxK09}.

\emph{Bit commitment} is another two-party primitive in which Alice wishes to commit to a specific bit value to Bob in such a way that Bob learns nothing about the committed value until Alice chooses to reveal it, yet Alice cannot reveal a value different from her commitment.
Any quantum protocol for bit commitment allows one party to cheat with probability at least $0.739$, and this bound is optimal \cite{ChaillouxK11}.
Our optimal bound of $2/3$ for semi-honest oblivious transfer adds yet another universal constant to the above list of cheating probabilities.
 
The first quantum protocol for oblivious transfer, called ``multiplexing'' at the time, was presented in a paper by Wiesner in the 1970's, which took until 1983 to get published  \cite{Wiesner83}. Wiesner observed that the security of his protocol rested upon technological limitations, and that his protocol is broken in an information theoretic sense.

Shortly thereafter, Bennett, Brassard, Breidbard, and Wiesner~\cite{BennettBBW83} presented another protocol wherein Alice can learn either bit with probability $\cos^2(\pi/8)\approx 0.854$ and the other bit is hidden afterwards. This protocol differs from our definition by sacrificing completeness in exchange for soundness, but it has other desirable properties such as its use for succinct random access codes \cite{Nayak99,AmbainisN+02}.

In 1997, it was shown by Lo~\cite{Lo97} that if Bob has no information about Alice's choice bit then Alice can learn both of Bob's data bits with certainty, rendering ideal quantum oblivious transfer impossible.
Since then interesting protocols that assume bounds on Alice's ability to cheat have been proposed, such as bounded quantum storage~\cite{DamgaardFSS08} or noisy quantum storage \cite{WehnerST08,Schaffner10}.
Recently, oblivious transfer in the noisy-storage model has been implemented in the laboratory \cite{ErvenN+13}.

There have been other analyses of the security of quantum oblivious transfer protocols. For example, Salvail, Schaffner and Sotakova \cite{SalvailSS09} give lower bounds on the amount of information that is leaked to a dishonest party in any oblivious transfer protocol.
In another work, Jain, Radhakrishnan and Sen \cite{JainRS09} showed a tradeoff between how many bits of information each player gets about the other party's input for $1$-out-of-$n$ oblivious transfer.
The security analyses of these works involves information notions and entropy and does not directly translate to cheating probabilities, which is the measure of security used in this paper.

\subsection{Mathematical preliminaries and notation}
 
We assume familiarity with quantum information; the purpose of this section is to clarify notation.
The \emph{trace norm} $\tnorm{X}$ of an operator $X$ is equal to the sum of the singular values of $X$.
The quantity ${\tnorm{\rho-\xi}}$ quantifies the observable difference between two quantum states $\rho,\xi$ in the sense that the maximum probability with which one could correctly identify one of $\set{\rho,\xi}$ chosen uniformly at random is
\[ \frac{1}{2} + \frac{1}{4}\tnorm{\rho-\xi}. \]
The \emph{fidelity} between two states $\rho,\xi$ is defined as
\[ \F(\rho,\xi) = \Tnorm{\sqrt{\rho}\sqrt{\xi}}. \]
\emph{Uhlmann's Theorem} asserts that for any states $\rho,\xi$ of system $\cX$ and any purifications $\ket{\phi},\ket{\psi}\in\cX\ot\cY$ of $\rho,\xi$ we have
\[ \F(\rho,\xi) = \max_U \abs{\bra{\phi} (I_\cX\ot U)\ket{\psi}} \]
where the maximum is over all unitaries $U$ acting on $\cY$.
The fidelity and trace norm are related by the Fuchs-van de Graaf inequalities \cite{FuchsvdG99}, which assert the following for any quantum states $\rho, \xi$:
\begin{equation}
  \label{eq:FvD-states}
  1 - \frac{1}{2}\tnorm{\rho - \xi} \leq \F(\rho,\xi) \leq \sqrt{1 - \frac{1}{4}\tnorm{\rho-\xi}^2}.
\end{equation}
We require only the first of these two inequalities.

\section{Lower bound on the security tradeoff} \label{sect:LowerBound}
 
In this section we prove Theorem \ref{thm:LowerBound} (Lower bound curve) by constructing cheating strategies for both Alice and Bob for any oblivious transfer protocol.
Our cheating strategies both mimic honest strategies until the end of the protocol.

Fix an arbitrary quantum protocol for oblivious transfer.
For each choice of $a\in\zo$ for Alice and $x_0,x_1\in\zo$ for Bob let
\[ \ket{\psi_{a, x_0, x_1}} \in \cA \ot \cB \]
denote the pure state of the entire system at the end of the protocol, assuming all measurements have been deferred and when both parties have been honest.
Here $\cA,\cB$ denote the spaces associated with Alice's and Bob's portions of that system, respectively.
Let
\[ \rho_{a,x_0,x_1} \defeq \Ptr{\cB}{\kb{\psi_{a,x_0,x_1}}} \]
denote the reduced state of Alice's portion of the system.
 
\subsection{Cheating Alice}
  
We now describe for each choice of $a\in\set{0,1}$ a strategy for cheating-Alice that begins by employing the strategy for honest-Alice in order to learn $x_a$ with certainty and then attempts to learn something about $x_\ol{a}$.

First, observe that since honest-Alice can learn $x_a$ with certainty, there must be a non-destructive measurement that allows her to do so without disturbing the state of the system.
We may assume without loss of generality that honest-Alice performs such a measurement, so that the reduced state of Alice's portion of the system after she has learned $x_a$ is still $\rho_{a,x_0,x_1}$.

Cheating-Alice can now learn something about the other bit $x_\ol{a}$ by performing the Helstrom measurement in order to optimally distinguish which of the two possible states she holds.
Since honest-Bob's data bits $x_0,x_1$ are uniformly random, in the case $a=0$ this strategy allows cheating-Alice to guess $x_1$ with probability
\[ \dfrac{1}{2} + \dfrac{1}{4} \Tnorm{\rho_{0,x_0,0} - \rho_{0,x_0,1}} \]
for each fixed choice of $x_0$.
Similarly, in the case $a=1$ this strategy allows her to guess $x_0$ with probability
\[ \dfrac{1}{2} + \dfrac{1}{4} \Tnorm{\rho_{1,0,x_1} - \rho_{1,1,x_1}} \]
for each fixed choice of $x_1$.
Our strategy for cheating-Alice calls for her to implement one of these two strategies at random, in which case she successfully cheats with probability
\[ \dfrac{1}{2} + \dfrac{1}{8} \Pa{\Tnorm{\rho_{0,x_0,0} - \rho_{0,x_0,1}} + \Tnorm{\rho_{1,0,x_1} - \rho_{1,1,x_1}}} \]
for each fixed choice of $x_0,x_1$, from which we conclude the following.

\begin{proposition}[Cheating probability for Alice]
\label{lemma:AliceCheating}

  In any {semi-honest quantum oblivious transfer protocol with perfect completeness,}  Alice can cheat with probability at least
  \[ 
    \half + \frac{1}{16} \Pa{
      \sum_{x_0\in\set{0,1}} \Tnorm{\rho_{0,x_0,0} - \rho_{0,x_0,1}} +
      \sum_{x_1\in\set{0,1}} \Tnorm{\rho_{1,0,x_1} - \rho_{1,1,x_1}}
    }
  \]
  where $\rho_{a,x_0,x_1}$ denotes the reduced state of Alice's portion of the system when both parties are honest and Bob has data bits $x_0,x_1$ and Alice has choice bit $a$.

\end{proposition}
 
\subsection{Cheating Bob}
  
Our strategy for cheating-Bob calls for him to implement a ``purification'' of a strategy for honest-Bob with uniformly random data bits.
In other words, he implements a uniform superposition over $x_0,x_1$ of honest strategies.
In order to do so, he requires two additional private qubits, which we associate with spaces $\cX_0,\cX_1$.
Conditioned on honest-Alice's choice $a$, the pure state $\ket{\xi_a}$ of the entire system after an interaction with honest-Alice is thus
\[ \ket{\xi_a} = \dfrac{1}{2} \sum_{x_0,x_1\in\set{0,1}} \ket{\psi_{a,x_0,x_1}} \ket{x_0} \ket{x_1} \in \cA\ot\cB\ot\cX_0\ot\cX_1. \]

We now describe two cheating strategies for Bob to attempt to learn $a$, one for each fixed choice of $s\in\zo$.
The strategy is to apply the unitary $\cont{U_s}$ acting on $\cB \otimes \cX_0 \otimes \cX_1$ specified below and then measure the $\cX_s$ register in the $\set{\ket{\pm}}$ basis and guess $a=s$ on outcome `$-$' and $a=\ol{s}$ on outcome `$+$'.

Intuitively, the unitary $\cont{U_s}$ is a controlled-unitary that tries to make the state of the system look as though the bit $x_s$ were equal to zero under the assumption that Alice chose $a=\ol{s}$.
Formally, the nontrivial actions of the unitaries $\cont{U_s}$ are specified by
\begin{alignat*}{2}
  \cont{U_0} &: \ket{\psi_{1,1,x_1}} \ket{1} \ket{x_1} \mapsto (I_{\cA} \otimes U_{0,x_1}) \ket{\psi_{1,1,x_1}} \ket{1} \ket{x_1} & \quad \textrm{for $x_1\in\zo$} \\
  \cont{U_1} &: \ket{\psi_{0,x_0,1}} \ket{x_0} \ket{1} \mapsto (I_{\cA} \otimes U_{1,x_0}) \ket{\psi_{0,x_0,1}} \ket{x_0} \ket{1} & \quad \textrm{for $x_0\in\zo$}
\end{alignat*}
where $U_{0,x_1},U_{1,x_0}$ are unitaries acting on $\cB$ satisfying
\begin{alignat*}{2}
  \F(\rho_{1,0,x_1},\rho_{1,1,x_1}) & = \bra{\psi_{1,0,x_1}}(I_\cA\ot U_{0,x_1})\ket{\psi_{1,1,x_1}} & \quad \textrm{for $x_1\in\zo$} \\
  \F(\rho_{0,x_0,0},\rho_{0,x_0,1}) & = \bra{\psi_{0,x_0,0}}(I_\cA\ot U_{1,x_0})\ket{\psi_{0,x_0,1}} & \quad \textrm{for $x_0\in\zo$}
\end{alignat*}
respectively, as per Uhlmann's Theorem.

Let us analyze cheating-Bob's probability of success in the case $s=0$; a similar analysis applies to the case $s=1$.
To this end we compute the squared norm of the projections of
\[ \cont{U_0}\ket{\xi_a} = \frac{1}{2} \Pa{ \sum_{x_1\in\zo} \ket{\psi_{a,0,x_1}}\ket{0}\ket{x_1} + \sum_{x_1\in\zo} U_{0,x_1}\ket{\psi_{a,1,x_1}}\ket{1}\ket{x_1} } \]
onto the states $\ket{\pm}\in\cX_0$ for each fixed choice of $a$.
After some calculations we find that this quantity is equal to
\[ \frac{1}{2} \pm \frac{1}{4}\sum_{x_1\in\zo} \Re\Pa{ \bra{\psi_{a,0,x_1}}(I_\cA\ot U_{0,x_1})\ket{\psi_{a,1,x_1}} }, \]
where $\Re(\cdot)$ denotes the real part of a complex number.
  
If $a=0$ then Alice must be able to learn $x_0$ with certainty, so that $\ket{\psi_{0,0,x_1}}$ and $\ket{\psi_{0,1,x_1}}$ are locally distinguishable by Alice and hence their inner product must equal zero regardless of the unitary $U_{0,x_0}$ applied to Bob's portion $\cB$.
In this case, the above quantity equals $1/2$ and so cheating-Bob correctly guesses $a=0$ with probability $1/2$.

On the other hand, if $a=1$ then
\[ \Re(\bra{\psi_{1,0,x_1}}(I_\cA\ot U_{0,x_1})\ket{\psi_{1,1,x_1}}) = \F(\rho_{1,0,x_1},\rho_{1,1,x_1}) \]
by our choice of $U_{0,x_1}$.
The probability of obtaining outcome `$\pm$' is therefore
\[ \frac{1}{2} \pm \frac{1}{4}\sum_{x_1\in\zo} \F(\rho_{1,0,x_1},\rho_{1,1,x_1}) \]
and so cheating-Bob correctly guesses $a=1$ with the larger of these two probabilities.

Thus, the probability over a uniformly random choice of $a$ that cheating-Bob correctly guesses Alice's bit $a$ conditioned on his choice $s=0$ is given by
\[ \frac{1}{2} + \frac{1}{8} \sum_{x_1\in\zo} \F(\rho_{1,0,x_1},\rho_{1,1,x_1}). \]
A similar argument establishes that the probability over a random $a$ that cheating-Bob correctly guesses Alice's bit $a$ conditioned on his choice $s=1$ is given by
\[ \frac{1}{2} + \frac{1}{8}\sum_{x_0\in\zo} \F(\rho_{0,x_0,0},\rho_{0,x_0,1}). \]
Our strategy for cheating-Bob calls for him to choose $s$ at random, from which we conclude the following.

\begin{proposition}[Cheating probability for Bob]
\label{lemma:BobCheating}

  In any {semi-honest quantum oblivious transfer protocol with perfect completeness,} Bob can cheat with probability at least
  \[ 
    \frac{1}{2} + \frac{1}{16} \Pa{ \sum_{x_0\in\zo} \F(\rho_{0,x_0,0},\rho_{0,x_0,1}) + \sum_{x_1\in\zo} \F(\rho_{1,0,x_1},\rho_{1,1,x_1}) }
  \]
  where $\rho_{a,x_0,x_1}$ denotes the reduced state of Alice's portion of the system when both parties are honest and Bob has data bits $x_0,x_1$ and Alice has choice bit $a$.

\end{proposition} 
 
\subsection{Obtaining the lower bound} \label{ssect:lowerbound}
  
It is straightforward to combine Propositions \ref{lemma:AliceCheating} and \ref{lemma:BobCheating} (Cheating strategies for Alice and Bob) in order to obtain a lower bound on the security of any quantum protocol for oblivious transfer.
To this end fix any such protocol and define the quantities
\begin{align*}
 F &\defeq \sum_{x_0\in\zo} \F(\rho_{0,x_0,0},\rho_{0,x_0,1}) + \sum_{x_1\in\zo} \F(\rho_{1,0,x_1},\rho_{1,1,x_1}) \\
 \Delta &\defeq \half\Pa{ \sum_{x_0\in\zo} \Tnorm{\rho_{0,x_0,0} - \rho_{0,x_0,1}} + \sum_{x_1\in\zo} \Tnorm{\rho_{1,0,x_1} - \rho_{1,1,x_1}} }
\end{align*}
Recall from Section \ref{sec:intro:weak-OT} that $\PA,\PB$ denote the maximum probabilities with which Alice and Bob can cheat, respectively.
By Propositions \ref{lemma:AliceCheating} and \ref{lemma:BobCheating} we have
\begin{align*}
  \PA &\geq \half + \frac{1}{8} \Delta \\
  \PB &\geq \half + \frac{1}{16} F.
\end{align*}
It follows immediately from the Fuchs-van de Graaf inequalities \eqref{eq:FvD-states} that $F + \Delta \geq 4$, from which we obtain the lower bound
\[ 2 P^\star_\mathrm{Bob} + P^\star_\mathrm{Alice} \geq 2, \]
completing the proof of Theorem~\ref{thm:LowerBound} (Lower bound curve).
 
\section{Protocols arbitrarily close to the tradeoff curve} \label{sect:UpperBound}
 
In this section we exhibit a family of protocols whose cheating probabilities can be made arbitrarily close to any point on the tradeoff curve of Theorem \ref{thm:LowerBound} (Lower bound curve).
To this end we present an optimal protocol for each of the two extremes of our lower bound curve: one in which Bob cannot cheat at all {(\textit{i.e.} $\PB = {1}/{2}$)} and another in which Alice cannot cheat at all {(\textit{i.e.}  $\PA = {1}/{2}$)}.
We then observe that any point on the curve can be approximated to arbitrary precision by playing one of these two protocols according to the outcome of an appropriately unbalanced weak coin flipping protocol.
 
\subsection{A trivial optimal protocol in which Bob cannot cheat} \label{ssect:AliceProtocol}
 
By definition, the condition that Bob cannot cheat means that $\PB = 1/2$.
It then follows from Theorem \ref{thm:LowerBound} (Lower bound curve) that $\PA \geq 1$ and hence Alice must be able to cheat perfectly in any such protocol.
(This special case of Theorem \ref{thm:LowerBound} was first observed by Lo \cite{Lo97}.)

There is a trivial protocol with these properties.
It is a non-interactive classical protocol in which Bob simply sends both of his bits $x_0,x_1$ to Alice, who selects the bit she wishes to learn.
It is remarkable that an optimal protocol for semi-honest oblivious transfer can be obtained by mixing with such a trivial protocol.
 
\subsection{An optimal protocol in which Alice cannot cheat} \label{ssect:BobProtocol}
 
By definition, the condition that Alice cannot cheat means that $\PA=1/2$.
It then follows from Theorem \ref{thm:LowerBound} (Lower bound curve) that $\PB \geq 3/4$.
It just so happens that the protocol of Chailloux, Kerenidis, and Sikora has these properties \cite{ChaillouxKS13-QIC}.
That protocol is reproduced below.

\begin{protocol}[Chailloux, Kerenidis, Sikora \cite{ChaillouxKS13-QIC}] \label{BobProtocol}

\quad
  \begin{enumerate}

  \item \label{it:CKS13:state}
  Alice randomly chooses index $a \in \zo$ and prepares the two-qutrit state $\ket{\psi_a} := \frac{1}{\sqrt 2} \ket{aa} + \frac{1}{\sqrt 2} \ket{22}$. She sends the second qutrit to Bob.

  \item
  Bob randomly chooses $x_0, x_1 \in \zo$ and applies the unitary $\ket{0} \mapsto (-1)^{x_0} \ket{0}$, $\ket{1} \mapsto (-1)^{x_1} \ket{1}$, $\ket{2} \mapsto \ket{2}$. He returns the qutrit to Alice. The state Alice has is
  \[ \frac{(-1)^{x_a}}{\sqrt 2} \ket{aa} + \frac{1}{\sqrt 2} \ket{22}. \]

  \item \label{it:CKS13:measure}
  Alice has a two-outcome measurement (depending on $a$) to learn $x_a$ with certainty.

  \end{enumerate}

\end{protocol}

Those authors observed that Bob's maximum cheating probability for Protocol \ref{BobProtocol} is $3/4$ and that Bob can achieve this cheating probability without Alice knowing he cheated.
However, their protocol was designed for a different definition of security against cheating-Alice: they showed that cheating-Alice can guess both of Bob's bits with probability $3/4$, but there is no guarantee that she learns one of the bits with certainty.
 
We now argue that the additional requirement that Alice learns one bit with certainty implies that Alice's maximum cheating probability for Protocol \ref{BobProtocol} drops to $1/2$, so that Alice cannot cheat at all in this protocol under our present security definition.
Indeed, this claim is a special case of the following result.

\begin{proposition}[Alice cannot cheat in Protocol \ref{BobProtocol}] \label{lemma:BobProtocol}

  Let $\delta\geq 0$ and suppose cheating-Alice prepares a state $\ket{\psi}$ in step \ref{it:CKS13:state} of Protocol \ref{BobProtocol} that allows her to guess $x_a$ with probability at least $1-\delta$ in step \ref{it:CKS13:measure}.
  Then $\ket{\psi}$ allows her to guess $x_\ol{a}$ with probability at most $\frac{1}{2} + \sqrt{\delta(1-\delta)} + \delta$. 
  By setting $\delta=0$ we obtain $\PA=1/2$ in Protocol \ref{BobProtocol}.
\end{proposition}

\begin{proof}
We prove only the case $a=0$, as the proof for $a=1$ is completely symmetric.
Any strategy for cheating-Alice in Protocol \ref{BobProtocol} begins by preparing a state $\ket{\psi}$ in step \ref{it:CKS13:state} of the form
\[ \ket{\psi} = \alpha \ket{e_0} \ket{0} + \beta \ket{e_1} \ket{1} + \gamma \ket{e_2} \ket{2} \]
where $\ket{e_0}$, $\ket{e_1}$, $\ket{e_2}$ are unit vectors (not necessarily orthogonal) and $|\alpha|^2 + |\beta|^2 + |\gamma|^2 = 1$.
After Bob applies his unitary, Alice has the state
\[ \ket{\psi_{x_0,x_1}} \defeq \alpha (-1)^{x_0} \ket{e_0} \ket{0} + \beta (-1)^{x_1} \ket{e_1} \ket{1} + \gamma \ket{e_2} \ket{2}. \]
Let $\rho_{x_0}$ denote the mixed state of Alice's system conditioned on Bob's first bit being $x_0\in\zo$, so that 
\[ \rho_{x_0} = \frac{1}{2} \Pa{ \altketbra{\psi_{x_0,0}} + \altketbra{\psi_{x_0,1}} }. \]
For brevity write $\ket{\vec{c}} = \ket{e_0}\ket{c}$ for $c = 0,1,2$ so that
\[
  \rho_{x_0} = \abs{\alpha}^2 \altketbra{\vec{0}} + \abs{\beta}^2 \altketbra{\vec{1}} + \abs{\gamma}^2 \altketbra{\vec{2}} + (-1)^{x_0} \alpha\gamma^* \ketbra{\vec{0}}{\vec{2}} + (-1)^{x_0} \alpha^*\gamma \ketbra{\vec{2}}{\vec{0}}.
\]
Then 
\[
  \Pr[\text{Alice correctly guesses $x_0$}] =
  \frac{1}{2} + \frac{1}{4} \Tnorm{\rho_0 - \rho_1} =
  \frac{1}{2} + \frac{1}{4} \Tnorm{2\alpha\gamma^* \ketbra{\vec{0}}{\vec{2}} + 2\alpha^*\gamma \ketbra{\vec{2}}{\vec{0}}} =
  \frac{1}{2} + \abs{\alpha\gamma}.
\]
A similar argument yields $\Pr[\text{Alice correctly guesses $x_1$}] = 1/2 + |\beta\gamma|$.
 
By assumption, cheating-Alice can guess $x_0$ with probability at least $1-\delta$ so that
\[
  \frac{1}{2} + \abs{\alpha\gamma} \geq 1 - \delta.
\]
Then
\[ 0 \leq \abs{\beta}^2 = 1 - \abs{\alpha}^2 - \abs{\gamma}^2 = 1 - (\abs{\alpha} - \abs{\gamma})^2 - 2\abs{\alpha\gamma} \leq 2\delta - (\abs{\alpha} - \abs{\gamma})^2 \]
from which we obtain the bounds
\begin{align}
  \abs{\beta}^2 &\leq 2\delta \label{eq:beta} \\
  (\abs{\alpha} - \abs{\gamma})^2 &\leq 2\delta. \label{eq:alpha-gamma}
\end{align}
We use \eqref{eq:alpha-gamma} to derive a bound on $\abs{\gamma}$.
We have
\[ \abs{\gamma}^2 = 1 -\abs{\alpha}^2-\abs{\beta}^2 \leq 1 - \Pa{ \abs{\gamma}-\sqrt{2\delta} }^2. \]
Solving this quadratic equation for $\abs{\gamma}$ we find
\begin{equation} \label{eq:gamma}
  \abs{\gamma} \leq \sqrt{\frac{1-\delta}{2}} + \sqrt{\frac{\delta}{2}}.
\end{equation}
Combining the bounds \eqref{eq:beta}, \eqref{eq:gamma} on $\beta,\gamma$, we obtain
\[
  \Pr[\text{Alice correctly guesses $x_1$}] =
  \frac{1}{2} + \abs{\beta\gamma} \leq
  \frac{1}{2} + \sqrt{2\delta}\cdot\Pa{\sqrt{\frac{1-\delta}{2}} + \sqrt{\frac{\delta}{2}}} = \frac{1}{2} + \sqrt{\delta(1-\delta)} + \delta
\]
as claimed.
\end{proof}
  
\subsection{Combining the protocols via weak coin flipping} \label{ssect:Protocol} 
 
As suggested at the beginning of this section, the two extreme protocols presented in Sections \ref{ssect:AliceProtocol} and \ref{ssect:BobProtocol} can be mixed according to the outcome of an unbalanced weak coin flipping protocol in order to yield a protocol that is arbitrarily close to any desired point on the tradeoff curve. 
In this subsection we explain informally how this protocol works.
In the subsequent subsection we present the technical details required to address concerns about the composability of the extreme protocols with weak coin flipping.

Mochon has shown that there exist quantum protocols for weak coin flipping with arbitrarily small bias \cite{Mochon07} (see also \cite{Pelchat13, ACGKM16}), meaning that for any $\epsilon>0$ and $\lambda=1/2$ we have
\begin{align*}
\Pr[\textup{The outcome is $0$ when Alice and Bob are honest}] &= \lambda, \\
\Pr[\textup{The outcome is $1$ when Alice and Bob are honest}] &= 1-\lambda, \\
\Pr[\textup{Cheating-Alice can force honest-Bob to output $0$}] &\leq \lambda + \epsilon, \\
\Pr[\textup{Cheating-Bob can force honest-Alice to output $1$}] &\leq 1-\lambda + \epsilon.
\end{align*}  
      
Chailloux and Kerenidis observed that Mochon's protocol can be repeated in order to yield a protocol with the above properties for any $\lambda\in[0,1]$ that can be expressed as an integer divided by a power of two \cite{ChaillouxK09}.
The set of such numbers is dense in the interval $[0,1]$, allowing us to design semi-honest oblivious transfer protocols approaching any point on the tradeoff curve.
Weak coin flipping protocols of this type shall be called \emph{$\lambda$-unbalanced}.

Let $m, n \in [1/2,1]$ be any numbers on the line $2n+m=2$ and expressible as an integer divided by a power of two.
For any $\epsilon>0$, we present a protocol for which Alice's and Bob's maximum cheating probabilities $\PA, \PB$ obey
\begin{align*}
\PA &\leq m+\epsilon/2 \\
\PB &\leq n+\epsilon/4
\end{align*}
so that $2 \PB + \PA \leq 2 + \epsilon$.
An informal statement of our final protocol follows.

\begin{protocol}[Optimal quantum semi-honest oblivious transfer protocol (informal)]
\label{prot:optimal}
\ \\
Let $\lambda \in [0,1]$ be such that $m=1/2+\lambda/2$ and $n=3/4-\lambda/4$.
  \begin{enumerate}
  \item Alice and Bob execute a $\lambda$-unbalanced protocol for weak coin flipping with bias $\epsilon$ to agree on a bit $c\in\zo$.
  \item If $c=0$, then execute the trivial oblivious transfer protocol of Section~\ref{ssect:AliceProtocol} with $\PB = 1/2$ and $\PA = 1$.
  \item If $c=1$, then execute Protocol~\ref{BobProtocol} with $\PA = 1/2$ and $\PB = 3/4$.
  \end{enumerate}
\end{protocol}

Intuitively, the security of Protocol \ref{prot:optimal} is implied by the following observations.
\begin{itemize}
\item
Alice can cheat more if she biases the coin flip toward $c=0$.
She can force this outcome with probability at most $\lambda+\epsilon$, in which case she can cheat with certainty.
Otherwise, she can cheat with probability $1/2$.
Therefore,
\[ \PA \leq (\lambda + \epsilon) + \Pa{1-\lambda - \epsilon}\cdot\frac{1}{2} = \frac{1}{2}+\frac{\lambda}{2}+\frac{\epsilon}{2} = m + \frac{\epsilon}{2}, \]
as claimed.
\item
Similarly, Bob can cheat more if he biases the coin flip toward $c=1$.
He can force this outcome with probability at most $1-\lambda+\epsilon$, in which case he can cheat with probability $3/4$.
Otherwise, he can cheat with probability $1/2$.
Therefore,
\[ \PB \leq \Pa{1-\lambda + \epsilon}\cdot\frac{3}{4} + \Pa{\lambda - \epsilon}\cdot\frac{1}{2} = \frac{3}{4}-\frac{\lambda}{4} + \frac{\epsilon}{4} = n + \dfrac{\epsilon}{4}, \]
as claimed.
\end{itemize}
  
\subsection{An optimal quantum protocol} 
 
We now present a more detailed description of our optimal protocol for semi-honest oblivious transfer so as to properly address concerns about the composability of weak coin flipping with other protocols.

\begin{protocol}[Optimal quantum semi-honest oblivious transfer protocol $(m,n,\epsilon)$]
\label{prot:fulloptimal}
\ \\
Let $\lambda \in [0,1]$ be such that $m=1/2+\lambda/2$ and $n=3/4-\lambda/4$.
  \begin{enumerate}
  \item Alice chooses $a$ and Bob chooses $(x_0, x_1)$ uniformly at random. (Actually, Alice could wait until the beginning of step \ref{it:Alice-start} to choose $a$ and Bob could wait until the beginning of step \ref{it:Bob-start} to choose $(x_0, x_1)$.)
  \item Alice and Bob execute a $\lambda$-unbalanced protocol for weak coin flipping with bias $\epsilon$ and do not measure their respective outcome registers $\cC_A$ and $\cC_B$ (leaving it in superposition). Then at the end of an honest protocol, they share the state
\[ \sqrt{\lambda} \ket{0}_{\cC_A} \ket{\xi_0}_{\cM_A \otimes \cM_B} \ket{0}_{\cC_B}
+ \sqrt{1 - \lambda} \ket{1}_{\cC_A} \ket{\xi_1}_{\cM_A \otimes \cM_B} \ket{1}_{\cC_B}, \]
where $\cM_A \otimes \cM_B$ are extra spaces for Alice and Bob, respectively. (We can assume the amplitudes are real by absorbing the phase into the definition of $\ket{\xi_0}$ and $\ket{\xi_1}$.) 
\item \label{it:Alice-start} Alice creates the two-qutrit state $\ket{\psi_a}_{\cA_1 \otimes \cA_2}$, where $\ket{\psi_a}$ is as defined in Protocol~\ref{BobProtocol}. Alice sends one qutrit to Bob:
\begin{itemize}
  \item Controlled on the outcome $c_A=0$, Alice sends a ``dummy'' qutrit $\ket{0}$.
  \item Controlled on the outcome $c_A=1$, Alice sends $\cA_2$.
\end{itemize}

\item \label{it:Bob-start} Bob creates the two-qutrit state $\ket{x_0, x_1}_{\cX_0 \otimes \cX_1}$. Bob sends two qutrits to Alice:
\begin{itemize}
  \item Controlled on the outcome $c_B=0$, Bob sends registers $\cX_0$ and $\cX_1$ to Alice.
  \item Controlled on the outcome $c_B=1$, Bob applies the unitary in Protocol~\ref{BobProtocol} to the qutrit sent by Alice and then returns it to her. He also sends a ``dummy'' qutrit $\ket{0}$ (so the message dimensions are consistent).
\end{itemize}
\item Alice and Bob measure to learn their outcomes of the weak coin flipping protocol. Alice measures to learn $x_a$:
\begin{itemize}
  \item If Alice's outcome $c_A$ is $0$, she measures $\cX_a$.
  \item If Alice's outcome $c_A$ is $1$, she measures $\cA_1 \otimes \cA_2$ as in Protocol~\ref{BobProtocol}.
\end{itemize}
\end{enumerate}
\end{protocol}

\begin{proposition}
Protocol~\ref{prot:fulloptimal} satisfies $\PA \le m + \eps/2$ and $\PB \le n + \eps/4$.
\end{proposition}

\begin{proof}
We now examine how Alice and Bob can cheat.

\paragraph{Cheating Alice}

If Alice cheats in the weak coin flipping protocol, they share a state of the form
\[ \sqrt{\tau} \ket{\xi'_0}_{\cM'_A \otimes \cM_B} \ket{0}_{\cC_B}  + \sqrt{1 - \tau} \ket{\xi'_1}_{\cM'_A \otimes \cM_B} \ket{1}_{\cC_B}, \]
where we have denoted $\cM'_A := \cC_A \otimes \cM_A$ here for brevity. Note that $\tau \leq \lambda + \epsilon$ by the definition of the weak coin flipping protocol.

Notice that Bob performs one of two operations after the weak coin flipping protocol subroutine, controlled on the value in the $\cC_B$ register which he then measures immediately after. Thus Alice's measurement statistics are not affected by whether Bob measures before or after the controlled-operation, so we assume he measures and obtains the bit $c_B$ which equals $0$ with probability $\tau$. Thus, with probability $1 - \tau$, Bob is going to perform the actions of Protocol~\ref{BobProtocol} and with probability $\tau$ he performs the actions of the trivial protocol.

After Alice's last message, call it the register $\cY$, they share a state
\[ \alpha \ket{e_0}_\spa{X}\ket{0}_{\spa{Y}} + \beta \ket{e_1}_\spa{X}\ket{1}_{\spa{Y}}  + \gamma \ket{e_2}_{\spa{X}} \ket{2}_{\spa{Y}} \]
which is independent of $x_0,x_1$ and the register $\spa{X}$ is shared between Alice and Bob. If $c_B = 0$, Alice can learn both bits with probability $1$ (assuming she knows $c_B = 0$).  If $c_B = 1$, then Alice can not learn any information about one bit after the other is learned with certainty, from Proposition~\ref{lemma:BobProtocol}. (Indeed Alice may have less  cheating power than guaranteed by Proposition~\ref{lemma:BobProtocol} since she does not have control over all of $\cX$.) Therefore, Alice can cheat with probability at most
\[ \tau + \half (1-\tau) = \half + \frac{\tau}{2} \leq \half + \frac{\lambda + \epsilon}{2} = m + \frac{\epsilon}{2}, \]
as desired.

\paragraph{Cheating Bob}
 
After the weak coin flipping protocol Alice and Bob share a state of the form
\[ \sqrt{1 - \tau} \ket{0}_{\cC_A} \ket{\xi'_0}_{\cM_A \otimes \cM'_B}  + \sqrt{\tau} \ket{1}_{\cC_A} \ket{\xi'_1}_{\cM_A \otimes \cM'_B}, \]
where we have denoted $\cM'_B := \cM_B \otimes \cC_B$ here for brevity. Note that $\tau \leq 1 - \lambda + \epsilon$ by the definition of the weak coin flipping protocol.

When $\cC_A$ is in the $\ket{0}$ state, Alice sends $\ket{0}$, when $\cC_A$ is in the $\ket{1}$ state, she sends the $\cA_2$ part of $\ket{\psi_a}$. Denote $\rho_a = \ptr{\cA_1}{\kb{\psi_a}}$.
Thus, after Alice's message, Bob's reduced state conditioned on the outcome $c_A$ given Alice measures $\cC_A$ is
\[ \mu_a := (1-\tau) \ptr{\cM_A}{\kb{\xi'_0}} \otimes \kb{0} + \tau \ptr{\cM_A}{\kb{\xi'_1}} \otimes \rho_a. \]
Since $a$ is chosen uniformly at random by Alice and since step \ref{it:Alice-start} is the only message from Alice with any information about $a$, Bob's optimal strategy is to perform the Helstrom measurement to try to learn $a$. Hence, Bob can cheat with probability
\[ \PB = \dfrac{1}{2} + \dfrac{1}{4} \tnorm{\mu_0 - \mu_1} = \dfrac{1}{2} + \dfrac{\tau}{4} \tnorm{\rho_0 - \rho_1} = \dfrac{1}{2} + \frac{\tau}{4} \leq \dfrac{1}{2} + \frac{1 - \lambda + \epsilon}{4} = n + \dfrac{\epsilon}{4}, \]
as desired.
\end{proof}
 
\section{Robustness of semi-honest oblivious transfer} \label{sec:crazy-sec}
 
The requirement that cheating-Alice guesses one of honest-Bob's input bits with certainty might seem overly demanding, as achieving certainty in the quantum world can be difficult if not impossible.
It is natural to wonder how our results are affected by relaxing this requirement so that cheating-Alice need only guess one of Bob's bits with probability $1-\delta$ for some $\delta\in[0,1/2]$.
A good security definition ought to be robust in the face of tiny perturbations such as this; in this subsection we provide some evidence in favour of semi-honest oblivious transfer.

Let us give this informal concept a name: \emph{$\delta$-robustness of semi-honest oblivious transfer}.
There are several distinct ways to formalize this concept, all of which are equivalent when $\delta=0$.
For example:
\begin{enumerate}

  \item \label{it:sec:condition}
    What is the maximum probability $p_1$ with which cheating-Alice can guess $x_\ol{a}$ conditioned on correctly guessing $x_a$?

  \item \label{it:sec:smarter}
    What is the maximum probability $p_2$ which which cheating-Alice can produce a pair $(x_0,x_1)$ such that $x_a$ is correct with probability at least $1-\delta$ and $x_\ol{a}$ is correct with probability $p_2$?

  \item \label{it:sec:stupid}
    Suppose cheating-Alice plays a strategy that allows her to guess $x_a$ with probability at least ${1-\delta}$ except that her final measurement is selected to guess $x_\ol{a}$ instead of $x_a$.
    What is the maximum probability $p_3$ with which she can succeed?

\end{enumerate}
It is clear that $p_1\leq p_2\leq p_3$.
Definitions \ref{it:sec:condition} and \ref{it:sec:smarter} are quite reasonable. 
Definition \ref{it:sec:stupid} seems less reasonable, as cheating-Alice produces a guess for either $x_a$ or $x_\ol{a}$, but not both.
Presumably, if cheating Alice wished to guess $x_\ol{a}$ then she would not begin by playing a strategy optimized for learning $x_a$. 
However, an upper bound for $p_3$ immediately yields an upper bound on the more reasonable quantities $p_1,p_2$.
Proposition \ref{lemma:BobProtocol} can be used to bound $p_3$ (and hence also $p_1,p_2$) by a continuous function of $\delta$.

For example, if we require cheating-Alice to be able to guess $x_a$ with probability $0.99$ (so that $\delta=0.01$) then by executing Protocol \ref{prot:fulloptimal} with $\lambda=0.219$ we find that the maximum cheating probability for either party can be made arbitrarily close to $0.695$, as compared to the $2/3$ maximum when $\delta=0$.
Proposition \ref{lemma:BobProtocol} ceases to be useful for this purpose at $\delta\approx 0.0443$, at which point $\lambda=0$ and there is no need for weak coin flipping since both cheating probabilities are then equal to $3/4$ in Protocol~\ref{BobProtocol}.
 
Naturally, the lower bound curve $2\PB+\PA\geq 2$ of Theorem \ref{thm:LowerBound} can only increase as $\delta$ increases (regardless of how one formalizes $\delta$-robustness) as cheating-Alice may choose from a larger set of strategies.
We leave it as an open question to find optimal bounds for these or other definitions of $\delta$-robust semi-honest oblivious transfer.
  
\section*{Acknowledgements}
 
Research at the Perimeter Institute is supported by the Government of Canada through Industry Canada and by the Province of Ontario through the Ministry of Research and Innovation.
GG also acknowledges support from CryptoWorks21. 
JS acknowledges support from a Government of Canada NSERC Postdoctoral Fellowship, the French National Research Agency (ANR-09-JCJC-0067-01), and the European Union (ERC project QCC 306537). Research at the Centre for Quantum Technologies at the National University of Singapore is partially funded by the Singapore Ministry of Education and the National Research Foundation, also through the Tier 3 Grant ``Random numbers from quantum processes,'' (MOE2012-T3-1-009).

\bibliography{bibliography}
\bibliographystyle{alpha}

\end{document}